\documentclass{rQUF2e}

\RequirePackage{amsfonts}
\RequirePackage{textcomp}
\RequirePackage{graphicx}
\RequirePackage{amsmath}
\RequirePackage{url}
\RequirePackage{hyperref}
\RequirePackage{xcolor} 
\RequirePackage{amsthm} 

\theoremstyle{plain}
\newtheorem{theorem}{Theorem}[section]
\newtheorem{corollary}[theorem]{Corollary}
\newtheorem{lemma}[theorem]{Lemma}
\newtheorem{proposition}[theorem]{Proposition}
\newtheorem{example}[theorem]{Example}

\theoremstyle{definition}

\theoremstyle{remark}
\newtheorem{remark}{Remark}

\date{}

\DeclareRobustCommand{\calA}[0]{{\cal A}}

\DeclareRobustCommand{\calD}[0]{{\cal D}}

\DeclareRobustCommand{\calF}[0]{{\cal F}}

\DeclareRobustCommand{\calN}[0]{{\cal N}}

\DeclareRobustCommand{\C}[0]{\mathbb{C}}

\DeclareRobustCommand{\E}[0]{\mathbb{E}}
\DeclareRobustCommand{\F}[0]{\mathbb{F}}

\DeclareRobustCommand{\L}[0]{\mathbb{L}}

\DeclareRobustCommand{\N}[0]{\mathbb{N}}
\DeclareRobustCommand{\O}[0]{\mathbb{O}}
\DeclareRobustCommand{\P}[0]{\mathbb{P}}
\DeclareRobustCommand{\Q}[0]{\mathbb{Q}}
\DeclareRobustCommand{\R}[0]{\mathbb{R}}

\DeclareRobustCommand{\U}[0]{\mathbb{U}}

\DeclareRobustCommand{\Vega}{\mathrm{V}}

\DeclareRobustCommand{\supstack}[2]{\stackrel{\mbox{{\scriptsize ${#1}$}}}{ {#2} }}     

\DeclareRobustCommand{\half}{\frac12}

\DeclareRobustCommand{\IND}{\mbox{\textsf{\sl 1}}}
\DeclareRobustCommand{\1}[1][ ]{\ensuremath{\IND_{#1}\,}}

\DeclareRobustCommand{\Lin}[1]{\mathrm{lin}\ }

\DeclareRobustCommand{\EX}[2][\E]{\ensuremath{{#1}\left[\, {#2}\, \right]}}

\DeclareRobustCommand{\sig}{\ensuremath{\sigma}}           

\DeclareRobustCommand{\eqlabel}[1]{\label{eq:#1}}                                     
\DeclareRobustCommand{\eq}[1]{\begin{equation}\eqlabel{#1}}                           
\DeclareRobustCommand{\eqx}{\begin{equation}}                                         
\DeclareRobustCommand{\eqend}{\end{equation}}                                         
\DeclareRobustCommand{\eqref}[1]{\mbox{(\ref{eq:#1})}}                                

\DeclareRobustCommand{\eqary}{\begin{eqnarray*}}
\DeclareRobustCommand{\eqaryend}{\end{eqnarray*}}

\DeclareRobustCommand{\gains}[1][a]{{#1\star DH_T}}                                     

\usepackage[style=numeric, backend=bibtex]{biblatex}
\begin{document}


\title{Deep Hedging:  Learning to Remove the Drift under Trading Frictions with Minimal Equivalent Near-Martingale Measures}

\author{H. BUEHLER$\dag$
P. MURRAY$^\ast$$\dag$${\ddag}$\thanks{$^\ast$Corresponding author Email: phillip.murray@jpmorgan.com} M. S. PAKKANEN${\ddag}$ and B. WOOD$\dag$\\
\affil{$\dag$JP Morgan\\
$\ddag$Imperial College London}}

\maketitle

\begin{abstract}
We present a machine learning 
approach for finding
minimal equivalent martingale measures for markets  simulators of tradable instruments, e.g.~for a spot price and
options written on the same underlying. We extend our results to markets with frictions, in which case we find ``near-martingale
measures'' under which the prices of hedging instruments are martingales within their bid/ask spread.

By removing the drift, we are then able to learn using Deep Hedging a ``clean'' hedge for an exotic payoff which is not polluted by the trading strategy trying to make money from statistical arbitrage opportunities.
We correspondingly highlight the robustness of this hedge vs estimation
error of the original market simulator. We discuss applications to two market simulators.

\end{abstract}

\section{Introduction}

A long-standing challenge in quantitative finance is the development of \emph{market models} for the dynamics of tradable instruments such as a spot price and options thereon.
The classic approach to developing such models is to find model dynamics in a suitable
 parameter space under which the respective
risk-neutral drift could be computed somewhat efficiently, c.f.~for example 
\cite{SCHOEN}, \cite{WISSEL}, \cite{HJMLEVY} for the case of equity option markets.
With this approach, realistic dynamics or estimation of statistically valid parameters
are an afterthought.

This article proposes to reverse this process by starting out with training a realistic
model of the market under the statistical measure -- and then find an equivalent ``near-martingale'' measure
under which the drifts of tradable instruments are constrained by their marginal costs such that
there are no \emph{statistical arbitrage opportunities}, that is, trading strategies which produce positive expected gains. In the absence of trading costs, this means
finding an equivalent martingale measure to ``remove the drift''.

 Indeed, we will show that absence of statistical arbitrage under a given measure is equivalent to the conditional expectation of the returns under this measure
being constrained by their marginal bid/ask prices. This result is of independent interest.\\

The main motivation for the present work is the application of our \emph{Deep Hedging} algorithm to construct hedging strategies
for contingent claims by trading in hedging instruments which include derivatives such as options. When described first in~\cite{DH}, we relied on markets simulated with classic quantitative
finance models. In \cite{DHOPT} we proposed a method to build market simulators of options markets under the statistical measure. Under this measure, we will usually find \emph{statistical arbitrage} in the sense that an empty initial portfolio has positive value. This reflects the realities of historic data: at the time of writing the S\&P~500 had moved upwards over the last
ten years, giving a machine the impression that selling puts and being long the market is a winning strategy. However, naively exploiting this observation risks falling foul of the ``estimation error''
of the mean returns of our hedging instruments.
In the context of
hedging a portfolio of exotic derivatives, the presence of statistical arbitrage
is undesirable as an optimal strategy will be a combination
of a true hedge, and a strategy which does not depend on
our portfolio, but tries to take advantage of the opportunities seen the market. It
is therefore not robust against estimation error of said drifts. Hence we propose using the method presented here
to generate a ``clean'' hedge by removing the drift of the market to increase robustness against
errors in the estimation of returns of our hedging instruments.

While our examples focus on simulating equity option markets -- in this case amounting to a \emph{stochastic implied volatility} model -- our approach is by no means limited to the equities case. In fact, it is entirely model agnostic and can be applied to any market simulator which generates paths of
tradable instruments under the same numeraire which are free of classic arbitrage.

In particular, our approach can be applied to ``black box" neural network based simulators such as those using Generative Adversarial Networks (GANs) described in \cite{DHOPT} or Variational Autoencoders as in \cite{buehler2020data}. Such simulators use machine learning methods to generate realistic paths from the statistical measure, but clearly no analytic expression to describe the market dynamics can be written. Our method allows constructing an equivalent risk neutral measure through further applications of machine learning methods.

\subsection{Summary of our Approach}

Given  instrument returns $DH_t$ across discrete time steps $t\in\{0=t_0<\cdots<t_m\}$, and convex costs $c_t$ associated with trading $a_t$ units of each instrument, we propose using our \emph{Deep Hedging} algorithm introduced in~\cite{DH}
to find a trading strategy $a^*$ and cash amount $y^*$ which maximize the \emph{optimized certainty equivalent} of a utility~$u$,
\[
	a \longmapsto \sup_{y\in\R}\ \EX{ u\left(y+\sum_t a_t \cdot DH_t - c_t(a_t) \right) - y } \ .
\]
Let $m$ denote the ``marginal cost'' of trading in our market.
We then define an equivalent measure $\Q^*$ by setting
\[
	\frac{ d\Q^* }{ d\P } := u'\left( y^* + \sum_t a^*_t\cdot DH_t - m_t(a^*_t) \right) \ .
\]
Under $\Q^*$ the market has no \emph{statistical arbitrage opportunities} in the sense that there is no strategy~$a$ which has positive expected returns, i.e.~
\eq{supa}
	\sup_a\ \EX[\E^*]{ \sum_t a_t \cdot DH_t - c_t(a_t) } \leq 0 \ .
\eqend
In the absence of transaction costs, $\Q^*$ is an equivalent martingale measure. In the presence of transaction costs, we show that removing statistical arbitrage is equivalent to the measure being
 an equivalent \emph{near-martingale} measure, in the sense that the drift of all tradable instruments must be dominated by the transaction costs. Moreover, $\Q^*$ is minimal among all equivalent 
 (near-) martingale measures with respect to the
$\tilde u$-divergence from~$\P$ where $\tilde u$ is the Legendre-Fenchel transform of~$u$.

The key insight of our utility-based risk-neutral density construction is that it relies only on solving the optimization problem of finding $a^*$ and $c^*$, not 
on any particular dynamics for the market under the~$\P$ measure. Therefore, it lends itself 
to the application of modern machine learning methods.
As mentioned above, this is particularly useful in the case of removing statistical arbitrage from a ``black box" market simulator, such as the GAN based approach discussed in~\cite{DHOPT}.
Through the choice of utility function, we are able to control the risk neutral measure we construct.

We demonstrate the power of this approach with two examples of option market simulators 
for spot and a number of volatilities. Specifically,
we train a Vector Autoregressive (VAR) model of the form
\[
	dY_t = \left( B - A_1 Y_{t-1} - A_2 Y_{t-2} \right) dt + \Sigma dW_t 
\]
for a vector of log spot returns and log volatilities $Y_t = ( \log {S_t}/{S_{t-1}}, \log\sigma^1_t, \ldots, \log\sigma^K_t )'$, and also a neural network based GAN simulator, 
and then in both cases use our approach to construct the above measure such that the resulting spots and option prices are near-martingale, and free from statistical arbitrage.

\subsection{Related Work}

We are not aware of attempts to numerically solve for a risk-neutral density
with the approach discussed here, as an application to stochastic implied volatility or otherwise. To our knowledge, ours is the first practical approach for implementing general statistically trained market models under risk-neutral measures.

The classic approach to stochastic implied volatilities is via the route of identifying analytically a risk-neutral drift given
the other parameters of the specified model.
The first applicable results for a term structure of implied volatilities are due to~\cite{SCHOEN}. 
 The first viable approach to a full stochastic implied volatility model for an entire fixed strike and fixed maturity option surface
 was presented in\cite{WISSEL},
using as parametrization also \emph{discrete local volatilities}.
Wissel describes the required continuous time drift adjustment for a diffusion driving a grid of such discrete local volatilities as a function of the free parameters.
Unnaturally, in his approach the resulting spot diffusion takes only discrete values at the strikes of the options at each maturity date and the approach is limited to a set grid of options defined in cash strikes
and fixed maturities.

More recently, a number of works
have shown that when representing an option surface with a L\'{e}vy kernel we can derive suitable Heath-Jarrow-Morton conditions on the parameters of the diffusion of the L\'{e}vy kernel such that the resulting stock
price is arbitrage-free, c.f.~\cite{HJMLEVY} and the references therein. Simulation of the respective model requires solving the respective Fourier equations for the spot price and options at each step in the path.

\subsection{Outline}
The rest of the article is organised as follows. In Section \ref{frictionless} we describe the theoretical framework and introduce the key method for constructing a risk-neutral measure in the firctionless case, which is then extended to the case with market frictions in Section \ref{friction}. Then, in Section \ref{results} we describe some of the consequences of the approach from a practical persective, and in Section \ref{experiments} we provide numerical experiments demonstrating the effectiveness of the method in practice.

\section{Frictionless risk-neutral case} \label{frictionless}

Consider a discrete-time simulated financial market with finite time horizon where we trade over time steps $0=t_0<\cdots<t_m=T$ where $T$ is the maximum maturity of all tradable instruments. Fix a probability space $\Omega$ and a probability measure $\P$ under which the market is simulated, which we will refer to as the ``statistical" measure. For each $t \in \{ t_0,\ldots,t_m\}$, we denote by $s_t$ the \emph{state} of the market at time $t$, including relevant information from the past. The state represents all information available to us, including mid-prices of all tradable instruments, trading costs, restrictions and risk limits.

The sequence of states $(s_t)_{t=0,\ldots,T}$ generates a sequence $\F = (\calF_t)_{t=0,\ldots,T}$ of $\sig$-algebras forming a filtration. Being generative means
that any $\calF_t$-measurable function $f(\cdot)$ can be written as a function of $s_t$ as $f \equiv f(\cdot;s_t)$. 

To simplify notation, we stipulate that the total number of instruments at each timestep is always~$n$. Let $H_t^{(t)}=(H_t^{(t,1)},\ldots,H_t^{(t,n)})$ be the and $\R^n$-valued, $\F$-adapted stochastic process of mid-prices of the liquid instruments available to trade at $t$. As above, $H_t^{(t)}$ is a function of $s_t$, and we also assume that $H$ is in $L^1(\P)$. Note that $H_t$ can represent a wide class of instruments, including primary assets such as single equities, indices, and liquid options. 

For each instrument we observe at time $T$ a final mark-to-market mid-value $H_T^{(t,i)}$ which will usually be the sum of any cashflows along the path, and
which is also assumed to be a function of $s_T$. That means $s_T$ must contain sufficient information from the past along the path:
for example, if the $i$th instrument tradable at $t$ is a call option with relative strike $k_i$
and time-to-maturity $\tau_i \leq T-t$ on a spot price process $S_t$, then the final value of this $i$th instrument is the payoff on the path, $H_T^{(t,i)} = (S_{t+\tau_i}/S_t - k_i)^+$.
Whilst for simplicity, we assume that all options mature within the time horizon, we can easily extend our method
to the case where options are allowed to mature after $T$ by valuing them in~$T$ at mid-prices. 

We further assume that discounting rates, funding, dividends, and repo rates are zero. Extension to the case where they are non-zero and deterministic is straightforward. 

At each time step $t$ we may chose an \emph{action} $a_t \in \R^n$ to trade in the hedging instruments~$H_t^{(t)}$ based on the information available in the state $s_t$, i.e.~$a_t\equiv a(s_t)$. The $\R^n$-valued, $\F$-adapted stochastic process $a = (a_0,\ldots,a_{m-1})$ defines a \emph{trading strategy} over the time horizon. To ease notation, it will be useful to define $a^\pm_t := \max(0, \pm a_t)$ for an action $a_t$, where the operation is applied elementwise. We also use~$e^i$ to refer to the~$i$th unit vector.

We start with the frictionless case, where the actions are unconstrained, i.e. the set of \emph{admissible actions}~$\mathcal{A}_t$ is equal to~$\R^n$ for all~$t$. 
In this case the terminal \emph{gain} of implementing a trading strategy $a = (a_0,\ldots,a_{m-1})$ with $a_t\in\calA_t$ is  given by
 \eq{Ga}
 	\gains :=  \sum_{t=0}^{m-1} a_t  \cdot DH_t \ \ \ \mbox{with} \ \ \ DH_t := H^{(t)}_T - H^{(t)}_t  \ . 
 \eqend
 
Our slightly unusual notation of taking the performance of each instrument to maturity reflects our ambition to look at option market simulators for ``floating" implied volatility surfaces where
the observed financial instruments change from step to step. If the instruments tradable at each time step are, in fact, the same fixed strike and maturity instruments, then
 \eq{Gdelta}
 	\gains \equiv  \sum_{t=0}^{m-1} \delta_t  \cdot dH_t \ \ \ \mbox{with} \ \ \ dH_t := H_{t+1} - H_t \ ,
 \eqend
where $\delta_t := a_t + \delta_{t-1}$ starting with $\delta_{-1} := 0$.

\subsection{Optimized certainty equivalents}

In order to assess the performance of a trading strategy, we are looking for risk-adjusted measures of performance instead of plain expected return. 
We will focus on the following case: let~$u$ be a strictly concave, strictly increasing utility function $u$ which is $C^1$ and normalized to both $u(0)=0$ and $u'(0)=1$.\footnote{
Note that the normalization is a convenience which is always achievable: If $\tilde u$ is concave and strictly increasing, then $u(x) := (\tilde u(x) - \tilde u(0))/ \tilde u'(0)$ satisfies these assumptions.}
Examples of such utility functions are the \emph{adjusted mean-volatility} function $u(x) := ( 1+ \lambda x - \sqrt{1 + \lambda^2 x^2})/\lambda$ proposed in~\cite{VICKY}, or the exponential utility $u(x)=(1-e^{-\lambda x})/\lambda$. 
We make the further assumption that $u(\gains)\in L^1(\P)$ for all~$a$. (This condition is not met when $u$ is the exponential utility and the market contains a Black \& Scholes process with negative drift.\footnote{
	To see this, assume $m=1$, $n=1$, $H_0:=1$, and let $H_T:=\exp( (-\mu - \half \sigma^2) T + \sigma \sqrt{T} W )$ for positive~$\mu$ and~$\sigma$, 
	and~$W$ standard normal. Then $\E[ u( - (H_T-H_0)) ] = -\infty$.
})

For a given utility function, define now the \emph{optimized certainty equivalent} (OCE) of the expected utility, introduced in~\cite{BT} as
\eq{Uu}
	U(X) := \sup_{y\in\R} \Big\{  \EX{ u(y+X) } - y \Big\}  \ .
\eqend
The functional $U$ satisfies the following properties: 

\begin{enumerate}[(i)]
\item Monotone increasing: if $X \geq Y$ then $U(X) \geq U(Y)$. A better payoff leads to higher expected utility.
\item Concave: $U(\alpha X + (1 - \alpha) Y) \geq \alpha U(X) + (1-\alpha) U(Y)$ for $\alpha \in [0,1]$. Diversification leads  higher utility.
\item Cash-Invariant: $U(X+c) = U(X)+c$ for all $c\in\R$. Adding cash to a position increases its utility by the same amount.
\end{enumerate}

This above properties mean that $-U(X)$ is a convex risk measure. 
Note that the assumptions $u(0) = 0$ and $u'(0) = 1$ imply that $u(y) \leq y$ for all $y$ and hence $U(0) = 0$. Furthermore, $U$ is finite for all bounded variables~$X$, since by monotonicity we have $U(X) \leq U( \sup X) = U(0) + \sup X < \infty$. 

Cash invariance means that $U(X -U(X)) = 0$, i.e. $-U(X)$ is the minimum amount of cash that needs to be added to a position in order to make it acceptable, in the sense that $U(X + c) \geq 0$. The cash-invariance property of the OCE means in particular that optimizing $U$ does not depend on our initial wealth. 

A classic example of such an OCE measure is the case where~$u$ is the exponential utility with risk aversion level $\lambda$. In this case we
obtain the \emph{entropy}
\[
U_\lambda(X) = - \frac{1}{\lambda} \log \E \left[ e^{-\lambda X} \right] \ . 
\]

We now consider the application of the optimized certainty equivalent to the terminal gains of a trading strategy. To this end, define
\eq{f_eq}
	F(y,a) := \EX[\E]{ u\big(\,y+\gains\,\big) - y  }  \ .
\eqend
\begin{lemma}
Suppose the market exhibits classic arbitrage. Then no finite maximizers of $F$ exist.
\end{lemma}

\begin{proof}
Assume that $\tilde a$ is a classic arbitrage opportunity with~$\P[ \tilde a \star DH_T \geq 0] = 1$ and $\P[ A ] = p > 0$
for a set $A=\{ \tilde a\star DH_T \geq g \}$ for some $g>0$. 
Note that $\tilde a \star DH_T \geq g \1[A]$. Let $n\in\N$. Then
\[
U( n \tilde a \star DH_T) \geq \sup{}_y:\ \E[ \1[A] u(y+ng) - y ] = \sup{}_y:\ p u(y+ng) - y = (*) \ .
\]
The last term is optimized by~$\tilde y = u'{}^{-1}(1/p)-ng$ as it solves~$p u'(\tilde y+ng) = 1$. Therefore,
$(*) = p u(u'{}^{-1}(1/p))-u'{}^{-1}(1/p)+ng$ which tends to infinity as~$n\uparrow\infty$.
Hence, no finite maximizer of~\eqref{f_eq} exists.
\end{proof}

\subsection{Utility-based risk neutral densities}

In the absence of classic arbitrage, we are now able to use the optimized certainty equivalent framework to 
construct an equivalent martingale measure. Before moving on to our main result, we will need the following lemma. 

\begin{lemma}\label{lem:derivexp}
	Let $f:\R \rightarrow \R$ be concave and $C^1$.  Assume that $f(\xi)\in L^1$ for all $\xi$ and that $\E[f(\xi^*)]\geq \E[f(\xi)]$ for some $\xi^*$.
	Then
	\eq{exchangederivs}
		\partial_\epsilon|_{\epsilon=0} \E[ f(\epsilon \xi + \xi^*)  ] =  \E[ f'(\xi^*) \xi ].
	\eqend
\end{lemma}

\begin{proof}
Define
$
	\Delta_\epsilon := \frac1\epsilon \left( f(\epsilon \xi + \xi^*) - f( \xi^*) \right) 
$
such that $\Delta_\epsilon \uparrow  \partial_\epsilon|_{\epsilon=0} f(\epsilon \xi + \xi^*) = f'(\xi^*) \xi$ since $f$ is concave. As a difference
between two $L^1$ variables $\Delta_\epsilon \in L^1$. Since $\xi^*$ maximizes the expectation of $f$, we also have $\E[\Delta_\epsilon] \leq 0$.
Using the dominated convergence shows that $\Delta_\epsilon \uparrow   f'(\xi^*) \xi \in L^1$ and therefore
that taking expectations and derivatives in~\eqref{exchangederivs} can be exchanged.
\end{proof}

Now we give the main result allowing the construction of utility-based equivalent martingale measures.

\begin{proposition}
Let $y^*$ and $a^*$ be finite  maximizers of
\eq{max_me2}
	y, a \longmapsto F(y,a) := \EX[\E]{ u\big(\,y+\gains\,\big) - y  }  \ .
\eqend
Then,
\eq{Dcrm}
	D^*  := u'\big(\, y^* + \gains[a^*]\, \big) 
\eqend
is an equivalent martingale density, i.e.~the measure $\Q^*$ defined via $d\Q^*:=D^* d\P$ is an equivalent martingale measure.\footnote{
We note that if $u$ is not strictly increasing, then $D^*$ is an absolutely continuous, but possibly not equivalent density. An example
is the CVaR ``utility''$u(x) = \min\{ x, 0 \}/(1-\alpha)$.}

\end{proposition}

\begin{proof}
We follow broadly the discussion in Section~3.1 of \cite{FS}.\\

\noindent
\emph{Show that $D^* (\gains)\in L^1$ with zero expectation:}
optimality of $c^*$ and $a^*$ imply first 
$0 = \partial_\epsilon\big|_{\epsilon=0} F(y^*,\epsilon a +  a^*)$.
Secondly, lemma~\ref{lem:derivexp} 
shows for arbitrary~$a$ that $D^*(\gains\in L^1)$ 
with
\eq{EDstarGa}
0 = \partial_\epsilon\big|_{\epsilon=0} F(y^*,\epsilon a +  a^*) = \EX{ D^*\ \gains } \ .
\eqend	 
				
\noindent
\emph{Show that $\E_t[D^* DH_t]=0$:}		
for the previous statement, set $a:=(0,\ldots,\1[A_t]e^i,\ldots,0)$ 
where $A_t$ is $\calF_t$-measurable, and where $e^i$ denotes the $i$th
unit vector.
We obtain
\eqx
	0 = \E[D^* DH^i_t | \calF_t ] \ .
\eqend

\noindent
\emph{Show that $D^* \in L^1$:}
recall that $\gains = \mbox{$\sum_{it}$} a^i_t\cdot DH^i_t $.
Since $u$ is concave and strictly increasing, $u'$ is decreasing and positive.
Then,
\eqary
	0  \leq u'( y^* + \gains[a^*] ) & \leq &
		u'\left( y^*+ \gains[a^*] \right)|DH| \1[ |DH| > 1 ]  \\
	&& + u'\left(y^* - \mbox{ess $\sup_{it}$} |a^*{}^i_t| \right)\1[ |DH| \leq  1 ] \in L^1
\eqaryend
since $y^*$ and $a^*$ were assumed to be finite, and since the previous step
with $a=1$ implies $u'\left( y^*+ \gains[a^*] \right)|DH|\in L^1$.

\noindent
\emph{Positivity of $D^*$:} since $u'$ is decreasing and positive
we have $\lim_{n\uparrow \infty} u'(n)\geq 0$
and therefore
$
	\P[ u'(\gains[a^*]) > 0 ]
	= 
	\lim_{n\uparrow 0} 
	\P[ u'(\gains[a^*]) > u'(n) ]
	=
	\lim_{n\uparrow 0} 
	\P[ \gains[a^*] \leq  n ]
	= \P[ \gains[a^*] < \infty] = 1
$, since $a^*$ being almost surely finite implies that $\gains[a^*]\in L^1$.

\noindent
\emph{$D^*$ has unit expectation:}
optimality of $y^*$ and $a^*$ implies
				\eq{Dst_E1}
					0 = \partial_y\big|_{y=y^*} F(y,a^*) = \EX{ D^* } - 1
				\eqend	
and therefore that $\E[ D^* ] = 1$.

\end{proof}
The density $D^*$ provides an equivalent martingale density. It is minimal among all equivalent martingale densities in the following sense:
the Legendre-Fenchel transform of the convex function $f(x):=-u(-x)$ is defined as.
\[
\tilde u(y) := \sup_{x \in \R} \{ yx - f(x) \} \ .
\]
The associated $\tilde u$-divergence between two distributions $\Q$ and $\P$ with $\Q \ll \P$ 
is then
\[
D_f (\Q | \P) = \E \left[ \tilde u \left( \frac{d\Q}{d\P} \right) \right] \ .
\]
It is a non-symmetric measure of the similarity between two probability distributions.

\begin{corollary} 
Let $
 \tilde u(y)
$ be the Legrendre-Fenchel transform of~$u$, and define $D^*$ as in \eqref{Dcrm}. Then, $D^*$ is a minimizer of the $\tilde u$-divergence 
\eq{divergence}
	D \longmapsto \EX[\E]{ \tilde u\left( D) \right) }
\eqend
over all equivalent martingale densities.
\end{corollary}

\begin{proof}

	The Legrende-Fenchel transform of the convex function $-u(-x)$ is 
	$
		\tilde u(y) = \sup_x ( yx + u(-x) ) = \sup_x ( u(x) - yx )
	$ which implies that for all $x$, 
	\eq{tildeuinq}
		\tilde u(y) \geq u(x) - yx \ .
	\eqend	
	Let $\calD_e:=\{\ D>0 :\ \E[D]=1,\, \E[D\,(\gains)] = 0 \ \mbox{for all $a$}\ \}$ be the set
	of equivalent martingale densities. Equation~\eqref{tildeuinq} implies
	for $y\rightarrow D\in \calD_e$ and $x\rightarrow c+\gains$, 
	\eq{ineqD}
		\inf_{D\in\calD_e} \EX{ \tilde u(D) } \geq \sup_{c,a} \Big\{\EX{ u\big(c + \gains\big) }  
				- c \Big\} = F(y^*,a^*) \ .
	\eqend
	Let $I= u'{}^{-1}(\R) \subseteq (0,\infty)$.
	For a given~$y\in I$ the sup in $\tilde u(y) = \sup_x ( u(x) - yx )$ is attained by
	$x = u'{}^{-1}(y)$ which yields
	\eq{utf_eq2}
		\tilde u(y) = u\big( u'{}^{-1}(y) \big) - y\ u'{}^{-1}(y)
	\eqend
	for all $y\in I$ and all $x\in\R$ as claimed above.
	Applying~\eqref{utf_eq2} to $\tilde u(D^*)$ yields that equality of both sides of~\eqref{ineqD}, proving our claim
	that $D^*$ is indeed a minimizer of~\eqref{divergence}.

\end{proof}

Thus, finding the $\tilde u$-minimal equivalent martingale measures is the dual problem of maximizing the expected utility. The key observation is that we now have a numerically efficient method to solving the primal problem via the application of machine learning methods.

In the case of the exponential utility, the $\tilde u$-divergence is the \emph{relative entropy} of $\Q$ with respect to $\P$, 
\[
H(\Q | \P) = \E \left[ \frac{d\Q}{d\P} \log \frac{d\Q}{d\P} \right] \ .
\]
The measure ~$\Q^*$ is the \emph{minimal entropy martingale measure} (MEMM) introduced by~\cite{FRITTELLI}, given by
\[
\frac{d\Q^*}{d\P} = \frac{ e^{-\gains[a^*]}}{\E [ e^{-\gains[a^*]}]}
\]
The measure is unique due to the strict convexity of the function $\tilde u(y) = y \log y$.

\begin{remark}
In the case where the returns are normally distributed, and the utility is the exponential utility, then the optimization is easily shown to be equivalent to solving the classic mean-variance objective of Markowitz \cite{MV} $U(X) = \E[X] - 1/2 \lambda Var[X]$, and in this case the found martingale measure removes the drift while preserving the covariance of the returns. 
\end{remark}

\subsubsection*{Direct Construction of Equivalent Martingale Measures}

An alternative to the above construction is described in \cite{FS} Section~3.1, as follows.
\begin{proposition}
Define ~$u$ as above and fix some initial wealth $w_0 \in \R $. Let $a^*$ be a finite maximizer of 
\eq{max_me1}
	a\longmapsto \EX[ \E ]{  u\big(\,w_0 + \gains\,\big) } \ .
\eqend
Then, the measure $\Q^*$ with density
\eq{simpleD}
	D^*  := \frac{  u'(w_0 + \gains[a^*])  }{ \EX[ \E ]{  u'(w_0 + \gains[a^*]) } } 
\eqend
is an equivalent martingale measure.
\end{proposition}

\begin{proof}
We prove that $\Q^*$ is a martingale measure.\footnote{We follow broadly section~3.1 in \cite{FS}.}

\noindent
\emph{Show $\E [ u'(w_0+\gains[a^*])\, \gains ]=0$ for all~$a$.}
For an arbitrary~$a$ we get
\eqx
	0 = \partial_\epsilon|_{\epsilon=0}  \E [  u(\,(w_0+\epsilon a+ a^*)\star DH_T\,) ]
	 \supstack{(*)}= \E [ u'(w_0+\gains[a^*])\, \gains ]  \ ,
\eqend
where $(*)$ follows from Lemma~\ref{lem:derivexp}.
Given that $a$ was arbitrary above also implies $\E [ u'(w_0+\gains[a^*])\, DH_t | \calF_t ] = 0$.

\noindent
\emph{We first prove that $u'(\gains[a^*]) \in L^1$:}
recall that $\gains = \mbox{$\sum_{it}$} a^i_t\cdot DH^i_t $.
Since $u$ is concave and increasing, $0\leq u'(x) \leq u'(x-\epsilon)$ for $\epsilon>0$.
Then,
\eqary
	0  \leq u'( w_0 + \gains[a^*] ) & \leq &
		u'\left(w_0 + \gains[a^*] \right)|DH| \1[ |DH| > 1 ]  \\
	&& + u'\left(- \mbox{ess $\sup_{it}$} |a^*{}^i_t| \right)\1[ |DH| \leq  1 ] \in L^1
\eqaryend
since $a^*$ was assumed to be finite.

\noindent
\emph{Positivity of $D^*$:} since $u'$ is decreasing and positive
we have $\lim_{n\uparrow \infty} u'(n)=0$
and therefore
$
	\P[ u'(w_0 + \gains[a^*]) > 0 ]
	= 
	\lim_{n\uparrow 0} 
	\P[ u'(w_0 + \gains[a^*]) > u'(n) ]
	=
	\lim_{n\uparrow 0} 
	\P[ w_0 + \gains[a^*] \leq  n ]
	= \P[ w_0 + \gains[a^*] < \infty] = 1
$, since $a^*$ being almost surely finite implies that $\gains[a^*]\in L^1$
\end{proof}

\begin{remark}
 We note that the assumption of finiteness of $a^*$ again excludes markets with classic arbitrage opportunities.\footnote{
 Assume that $\tilde a$ is a classic arbitrage opportunity with~$\P[ \tilde a \star DH_T \geq 0] = 1$ and $\P[ A ] = p > 0$
for a set $A=\{ \tilde a \star DH_T \geq g \}$ for $g>0$.
Then $\E[ u(w_0 +  ( n \tilde a 1_A) \star DH_T ) ] \geq \E[  1_A u(w_0 +  n g  ) ] + (1-p) u(w_0) \geq p\,u(w_0 + ng) + (1-p) u(w_0) \uparrow p\,u(\infty) + (1-p) u(w_0)$, e.g.~no finite
maximizer of $\E[ u(\gains) ]$ exists.
} If~$u$ is the exponential utility, then $D^*$ coincides with the
previously defined  density of the MEMM in~\eqref{Dcrm}.  
\end{remark}

We note that while this approach is somewhat more direct it depends on initial wealth -- except in the case of the exponential utility --
and lacks the interpretation of the 
density as a minimizer of some distance to~$\P$.

We now briefly discuss some extensions of the previous results to the cases of unbounded assets, and continuous time processes.

\subsubsection*{Unbounded Assets}

As pointed out in \cite{FS} Section~3.1, the requirement $u(\gains)\in L^1$ can  be enforced at the cost of interpretability
of our previous results by passing over to bounded asset prices:
to this end define the random variable $M:=\max_{i,t} | DH_t^i |$ and set $D\bar H^i_t := DH^i_t/(1+M)$, which are now bounded.

We can then show with the same steps as before that we can construct an equivalent martingale measure in this case as follows.

\begin{proposition}
Let $y^*$ and $a^*$ be maximizers of the bounded problem
\eqx
	y, a \longmapsto \bar F(y,a) := \EX{ 	u\left(\frac{ y + \gains }{ 1+M } \right) - y } \ .
\eqend
Then, 
\eq{dbounded}
	D^* := \frac{ u'\!\left(\frac{ y^*+ \gains[a^*] }{ 1+M } \right) }{ 1+M } 
\eqend
is an equivalent martingale density for the unscaled problem, i.e.~$\E^*[\gains]\leq 0$
for all~$a$.

Moreover, $D^*$ minimizes the scaled $\tilde u$-divergence
\eqx
	D \rightarrow \EX{ \tilde u\big( (1+M)D \big) } 
\eqend
over all equivalent martingale densities.
\end{proposition}

\begin{proof}
We cover the main differences to the previous case: first, we see that
\eqx
	0 = \partial_y F(y^*,a^*) = \EX{ \frac{  u'\left(\frac{ y^*+\gains[a^*] }{ 1 + M} \right) }{ 1+M } - 1 } \ .
\eqend
Then,
\eqx
	0 = \partial_\epsilon F(y^*,\epsilon a + a^*) = \EX{ \frac{  u'\left(\frac{ y^*+\gains[a^*] }{ 1 + M } \right) }{ 1+M } \gains  } 
\eqend
showing that $D^*$ is an equivalent martingale density.
Using~\eqref{tildeuinq} with $y\rightarrow (1+M)D$ for $D\in\calD_e$ and $z\rightarrow (c+\gains)/(1 + M)$
yields as before
	\eqx
		\inf_{D\in\calD_e} \EX{ \tilde g\big(\, (1 + M)D\,\big) } \geq \sup_{c,a} \Big\{\EX{u\left(\frac{ c+\gains }{ 1 + M } \right) }  
				- c \Big\} \ .
	\eqend
Equality in $D^*$ follows as before.
\end{proof}

\subsubsection*{Continuous Time}
	We note that our method of proof also works in a continuous time: let $G_t(a) := \int_0^t\!\!\,a_r\,dH_r$ where $dH_t = \mu_t\,dt + \sigma_t\,dW_t$ (i.e.~the classic setup
	with fixed instruments).
	Assume that $y^*,a^*$ maximize $y,a \mapsto \E[ u( y+G_T(a) ) -y ]$ and that $\E[ u(G_T(a^*) ) ]<\infty$ which again excludes markets with classic arbitrage.
	Then, then same statement as above is true with virtually the same proof.
	
	Let $G^*_t := G_t(a^*)$ and notice that if~$u\in\C^3$ then
	$f(x) := \log u'(x)$ has derivative $f'(x) = \frac{ u''(x) }{ u'(x) }$ which
	is the \emph{Arrow-Pratt coefficient of absolute risk aversion} of $u$, c.f.~\cite{FS} section~2.3.	 Standard calculus shows
	that
	\eqx
		D^*_t = \exp\left(
			 \int_0^t\!\! \frac{ u''( G^*_r) }{ u'(G^*_r) } \sigma_r\,dW_r - \half \int_0^t\!\! \left( \frac{ u''(G^*_r) }{ u'(G^*_r) } \sigma_r \right)^2\!\! dr
		\right)
	\eqend
	Under $\Q^*$ our assets are driftless and satisfy $dH_t = \sigma_t\,dW^*_t$ for a $\Q^*$-Brownian motion $W^*$.
	This implies the well-known result
	\eqx
		\frac{ u''(G^*_t) }{ u'(G^*_t) } = \frac{ \mu_t }{ \sigma^2_t } \ .
	\eqend

\section{Transaction costs and trading constraints} \label{friction}

The previous section enables us to simulate markets from a martingale measure in the absence of trading frictions. In
practise, trading strategies will be subject to trading cost and constraints such as liquidity and risk limits.
Our use-case is training a \emph{Deep Hedging} agent. We therefore now extend the 
previous results to the case of generalized cost functions which will cover both trading cost and most trading constraints.

A~\emph{generalized cost function} is a non-negative, $\mathcal{F}_t$-measurable function $c_t(a_t)\equiv c(a_t;s_t)$ with values in $[0,\infty]$, which is convex in $a_t$, lower semi-continuous, and normalized to $c_t(0)=0$. To impose convex restrictions on our trading activity, we set transaction cost to infinity outside the admissible set. Indeed, let $\calA_t$ be 
be a convex set of admissible trading actions, and $\bar c_t$ an initial const function. We then use $c_t(a_t) := \bar c_t(a_t) + \infty \1[ a_t\not \in\calA_t]$. (We note that this
construction is lower semi-continuous.) As example, let us assume the $i$th instrument is not tradable in~$t$. We then impose $c_t(a_t) = \infty$ whenever $|a^i_t|>0$.

In reverse, if $c_t$ is a generalized cost function, we may call $\calA_t := \{ a\in\R^n:\, c_t(a)<\infty\}$  the convex set of \emph{admissible actions}. 
Note also that by construction~$0 \in \calA_t$.

\begin{example}
	The simplest trading costs are proportional. Assume that $\Delta_t$ and $\Vega_t$ are observable Black~\& Scholes delta and vega of the mid-prices $H^{(t)}_t$ for 
	the trading instruments available at~$t$, and that the cost of trading $a^i$ units of $H^{t,i}_t$ is proportional to its delta and vega
	with cost factors $g^\pm_\Delta$ and $g^\pm_\Vega$ for buying and selling, respectively.
	We also impose that we may trade at most $\Vega_\mathrm{max}$ units of vega per time step.
	The corresponding cost function is given by
	\eqx
		c_t(a) := \left\{ \begin{array}{ll}
					 a^+ \cdot \left(  g^+_\Delta \Delta_t  + g^+_\Vega \Vega_t \right) &  a\cdot \Vega_t \leq \Vega_\mathrm{max}
					  \\
					 \ \ \ \ \ + a^- \cdot \left(  g^-_\Delta \Delta_t  + g^-_\Vega \Vega_t \right) &
					 \\
					 \infty &  a\cdot \Vega_t > \Vega_\mathrm{max}
					 \end{array}
					 \right.
	\eqend
\end{example}

\begin{example}\label{ex:constricted}
	Consider trading cost which apply only to net delta and vega traded, e.g.
	\eqx
		c_t(a) := \left\{ \begin{array}{ll}
					  g^+_\Delta \left( a \cdot \Delta_t \right)^+ + g^+_\Vega \left( a \cdot \Vega_t \right)^+ &  a\cdot \Vega_t \leq \Vega_\mathrm{max}
					  \\
					 \ \ \ \ \ + g^-_\Delta \left( a \cdot \Delta_t \right)^- + g^-_\Vega \left( a \cdot \Vega_t \right)^- &
					 \\
					 \infty &  a\cdot \Vega_t > \Vega_\mathrm{max}
					 \end{array}
					 \right.
	\eqend
\end{example}	

 The terminal gain of implementing a trading policy~$a$ with cost function~$c$
 is given by
 \eq{Ga2}
 	\gains - C_T(a)  \ \ \ \mbox{where} \ \ \ C_T(a) := \sum_{t=0}^{T-1} c_t(a_t) \ .
 \eqend

The marginal cost of trading small quantities of the $i$th asset in~$t$
are given as
\eq{def_marginal}
	\gamma^{i+}_t := + \partial_{\epsilon>0} c_t(\epsilon e^i)
	\ \ \ \mbox{and} \ \ \
	\gamma^{i-}_t := - \partial_{\epsilon>0} c_t(-\epsilon e^i) \ .
\eqend
They define the \emph{marginal cost function}
\eq{def_m}
	m_t(a) := a^+ \cdot \gamma^+_t  - a^- \cdot \gamma^-_t \ ,
\eqend

\subsection{Statistical arbitrage and near-martingale measures}

Under the statistical measure we expect there to be \emph{statistical arbitrage opportunities}, i.e.~trading strategies $a$ such
that we expect to make money:
\eqx
	 \E[ \gains ] > 0 \ .
\eqend

In the absence of transaction costs, the market will be free from statistical arbitrage if and only if we are under a martingale measure.\footnote{
Assume that there is a~$t$ such that $f_t := \E[ DH_t|\calF_t ] \not= 0$. Set $a_t := \mathrm{sign}\,f_t$. Then the strategy $a=(0,\ldots,a_t,\ldots,0)$
is a statistical arbitrage strategy.}
Since the gains of trading with transaction costs are almost surely never greater than the gains in the absence of transaction costs, it is clear that if $\Q$ is an equivalent martingale measure for the market, then there are no statistical arbitrage opportunities
under transaction cost, either, i.e.~$
	\E_{\Q}\left[\, \gains - C_T(a) \right] \leq 0
$
for all policies~$a$ (equality is acheived with $a\equiv0$). Taking the limit to small transaction cost, it becomes inutitively clear
that $\E_{\Q}\left[\, \gains - M_T(a) \right] \leq 0$ as well for marginal cost. In fact, inuitively it makes sense that the market is free of statistical arbitrage
with full cost~$c$ if and only if it is free of statistical arbitrage with marginal cost~$m$.

Here is our formal result:

\begin{proposition} \label{prop:no_stat_arb}
We call $\Q$ a \emph{near martingale measure} if any of the following equivalent conditions hold:
\begin{itemize}
\item
the measure $\Q$ is free from statistical arbitrage with full cost~$c$;

\item
the measure $\Q$ is free from statistical arbitrage with marginal cost~$m$; and

\item the expected return from any hedging instrument is within its marginal bid/ask spread in the sense that
\eq{pricelimits_ei}
		\underbrace{ H^{(t,i)}_t - \gamma^{i-}_t }_{\begin{array}{c}\mbox{Marginal}\\ \mbox{bid price} \end{array}}
		\leq 
		\underbrace{ \E_\Q\big[H^{(t,i)}_T \big| \mathcal{F}_t\big]}_{\mbox{Expected gains}}
		\leq 
		\underbrace{ H^{(t,i)}_t + \gamma^{i+}_t }_{\begin{array}{c}\mbox{Marginal}\\ \mbox{ask price}\end{array}} \ ,
\eqend 
with~$\gamma^{i\pm}_t$ defined in~\eqref{def_marginal}.
\end{itemize}
\end{proposition}

\begin{proof}
Assume first there are no statistical arbitrage opportunities with full cost~$c$. We will show~\eqref{pricelimits_ei}.
Let~$A\in\calF_t$ arbitrary
and let~$e^i_t$
the policy with unit vector~$e^i$ at~$t$ and zero elsewhere; for ease of notation we will also write $e^i_t$
for simply the unit vector, seen at a time~$t$.

Absence of statistical arbitrage implies that
$
0 \geq  \frac1\epsilon \E_\Q\big[  \gains[(\pm \epsilon e^i_t 1_A)] 
  - C_T(\pm \epsilon e^i_t 1_A)
\big]$  for all $\epsilon>0$,
and therefore
\[
	0 \geq
			\partial_{\epsilon>0} \E_\Q\big[ 1_A \left\{ \pm \epsilon \gains[ e^i_t ] - C_T(\pm  \epsilon e^i_t) \right\} \big]
			=
				  \E_\Q\big[ 1_A\ \big\{ DH_t^i \mp \gamma^{i\pm}_t \big\} \big] 
\]
which yields~\eqref{pricelimits_ei}.

Assume now that~\eqref{pricelimits_ei} holds, and let $a$ be arbitrary.
Then
$
		\E_\Q[ \sum{}_t a_t \cdot  DH_t - m_t(a_t)]  \leq 0
$
by construction of our marginal cost~\eqref{def_m} and~\eqref{pricelimits_ei}. Hence, there is no statistical arbitrage with cost~$m$.
Since $c_t\geq m_t$ it is also clear that if there is no statistical arbitrage with marginal cost~$m$, then there is also no statistical
arbitrage with full cost~$c$.
\end{proof}

\begin{remark}
Under the conditions of the above theorem the conditional expectation $\E_{\Q} \big[H^{(t,i)}_T \big| \mathcal{F}_t\big]$ defines a martingale ``micro-price'' \cite{STOIKOV} within the bid--ask spread.
\end{remark}

In the absence of trading costs or trading constraints then equality is acheived. That is, the market is free from statistical arbitrage if and only if 
\[
\E_\Q\big[H^{(t,i)}_T \big| \mathcal{F}_t\big] = H^{(t,i)}_t \ .
\]
resulting in the classic formulation of the price process being a martingale under $\Q$.

\subsection{Utility-based near-martingale measures under trading frictions}

We now proceed with constructing a near-martingale measure~$\Q^*$ via the same duality as in the zero transaction cost case.
Define again the function 
\eq{max_me2tc}
	F(y,a) := \EX{ u\big( y + \gains - M_T(a) \big) - y }
\eqend
just as in~\eqref{max_me2}, but this time with marginal transaction costs. 

\begin{proposition}
Let $y^*$ and $a^*$ be finite maximizers of $y, a \mapsto F(y,a)$.  Then
\eqx
	D^* := u'\!\left( y^* + \gains[a^*]  -  M_T(a^*)  \right) 
\eqend
is an equivalent density, and the measure $\Q^*$ defined by $d\Q^*:=D^*d\P$ is a near-martingale measure.
Moreover, the density $D^*$ minimizes the $\tilde u$-divergence among all equivalent near-martingale densities.
\end{proposition}

\begin{proof}
To show that $D^*$ is a equivalent near- martingale density most of the previous proof
applies as before, except of course~\eqref{EDstarGa} since $D^*$ is not an equivalent 
martingale measure. Instead, we will show that there is no statistical arbitrage
under $\Q^*$. Let therefore $A \in \calF_t$ be arbitrary, and denote by~$e^i_t$ the
strategy with unit vector $e^i$ in~$t$ and zero otherwise; for notational simplicity we will
also use  $e^i_t$ to refer simply the unit vector, in~$t$.

Define $F_\pm(\epsilon) := \pm \EX[\E]{ u\big( y^* + \gains[( \pm \epsilon {\1[A]} e^i_t + a^*)] - M_T( \pm \epsilon \1[A] e^i_t + a^* ) \big) - y^* }$.
Consider the derivative $
	\partial_\epsilon F_\pm(0) = 
	\EX[\E^*]{
		\1[A] \left\{
			DH^i_t \pm \partial_{\epsilon>0} m_t(\pm \epsilon e^i_t + a^*_t )
		\right\}
	} 
$: 
we recall that $m_t(a) = a^+_t \cdot \gamma^+_t - a^-_t \gamma^-_t$. Therefore
\[
	(*) = \pm  \partial_{\epsilon>0} m_t(\pm \epsilon e^i_t + a^*_t ) =
		\left\{
			\begin{array}{ll}
						-   \gamma^{i+}_t & \mbox{if $a^{*i}_t > 0$,} \\
						+   \gamma^{i-}_t & \mbox{if $a^{*i}_t < 0$,} \\
						\mp \gamma^{i\pm}_t & \mbox{if $a^{*i}_t = 0$.} \\
			\end{array}
		\right.
\]
Since $(a^*, y^*)$ are optimal we must have $0\in [ \min (*), \max (*) ]$. Given that~$A\in\calF_t$ was arbitrary
we obtain
\eqx
		\left\{
			\begin{array}{rlll}
						                  & \E^*[ DH_t | \calF_t ] & = + \gamma^{i+}_t        & \mbox{if $a^{*i}_t > 0$,} \\
						- \gamma^{i\pm}_t \leq &  \E^*[ DH_t | \calF_t ] & \leq + \gamma^{i+}_t & \mbox{if $a^{*i}_t = 0$, and } \\
						 - \gamma^{i-}_t  = & \E^*[ DH_t | \calF_t ] &                         & \mbox{if $a^{*i}_t < 0$.} \\
			\end{array}
		\right.
\eqend
This is in fact a more precise statement than~\eqref{pricelimits_ei}.

We now show that $D^*$ minimizes the $\tilde u$-divergence among all measures $D\in\calD_e := 
\{\ D>0:\ \E[D]=1, \E[D\,( \gains -  M_T(a))]\leq 0\ \mbox{for all~$a$}\ \}$. We apply~\eqref{tildeuinq} again
with $y\rightarrow D\in\calD_e$ and $x\rightarrow y + \gains - M_T(a)$. This yields
\eqx
	\E[ \tilde u(D) ] \geq \E[ u(y+\gains - M_T(a))  ] - \E[ D\,(y + \gains  - M_T(a)) ] 
	                   \geq \E[ u(y+\gains - M_T(a)) - y] \ ,
\eqend
where the last inequality holds since $D$ does not admit statistical arbitrage.
The right hand side is maximized in~$(a^*,y^*$).
For the left hand side, apply
again~\eqref{utf_eq2} which yields
\eqx
\E[ \tilde u(D^*) ] =\E[ u(y^*+\gains[a^*] - M_T(a^*) ) - y^* ] - \underbrace{ \E^*[ \gains[a^*] - M_T(a^*)  ] }_{=0} \ . 
\eqend
This proves that $D^*$ is $\tilde u$-minimal among all near-martingale measures.
\end{proof}

Considering that any equivalent true martingale measure is also a near-martingale measure, this result is
a formalization of the intuitive notion that in order to avoid statistical arbitrage we do not truly have to 
find a full martingale measure, but that we only have to ``bend''
the drifts of our trading instruments enough to be dominated by prevailing trading cost.

\section{Learning to Simulate Risk-Neutral Dynamics} \label{results}

The key insight of our utility-based risk-neutral density construction is that it relies only on solving the optimization problem of find $a^*$ and $y^*$, not 
on specifying any particular dynamics for the market under the $\P$ measure. Therefore, it can be done in a data-driven, model agnostic way, lending itself 
to the application of modern machine learning methods. Specifically, given a set of $N$ samples from a $\P$ market simulator, we may make the sample set risk neutral by
numerically solving the optimization problem on the $N$ paths, and then using our formulation to reweight the paths so that the resulting weighted sample is a (near-)martingale.
As mentioned above, this is particularly useful in the case of removing statistical arbitrage from a ``black box" market simulator, such as the GAN based approach discussed in~\cite{DHOPT}.

Our approach enables the adaptation of GAN and other advanced machine learning approaches so that they can not only simulate realistic samples from the statistical measure, but also from an equivalent risk neutral measure. Moreover, through the choice of utility function, we are able to control the risk neutral measure we construct.

We solve the stochastic control problem \eqref{max_me2} through an application of the `Deep Hedging' methods of \cite{DH}: we can pose \eqref{max_me2} as a reinforcement learning problem and use a neural network to represent our trading policy $a$, and since the function $F$ is fully differentiable, use stochastic gradient methods to find $a^*, y^*$, and hence $D^*$.

\subsection{Deep Hedging under Risk-Neutral Dynamics}

Our primary application is in the pricing and hedging of exotic options via utility-based Deep Hedging. With a portfolio of derivatives represented by the random variable $Z$ to hedge, the Deep Hedging problem under the statistical measure~$\P$ is to maximize the optimized certainty equivalent
\eq{dh_max_me2tc}
	\U_\P (Z) := \sup_{a,y}\ \EX[\E_\P]{ u\big( y + Z + \gains - C_T(a) \big) - y }
\eqend
over strategies~$a$ and~$y\in\R$. An optimal solution $a^*_\P$ is called an optimal hedge for~$Z$. We note that in the presence of statistical arbitrage $\U^*(0)>0$.
Deep Hedging under a near-martingale measure~$\Q^*$ then is 
\eq{dhrn_max_me2tc}
	\U^*(Z) := \sup_{a,y}\ \E^* \left[ u\big( y + Z + \gains- C_T(a) \big) - y \right].
\eqend

Since under $\Q^*$ we have $\U^*(0)=0$ we note that $\U^*$ represents an \emph{indifference price for~$Z$} in the sense of~\cite{DH} section~3.

In the case of hedging under exponential utility with zero transaction costs, it is straightforward to show that the optimal hedge for the derivative $Z$ under the statistical measure can be written as
\eqx
 	a^*_\P = a_{\Q}^* + a_0^* \ ,
\eqend
where $a_{\Q}^*$ is an optimal hedge for $Z$ under the minimal entropy martingale measure (MEMM), and where $a_0^*$ is an optimal statistical arbitrage strategy, i.e.~an optimal ``hedge'' for an empty initial portfolio.
In this sense we may regard $a^*_\Q$ as a \emph{net hedging strategy} for~$Z$.

A~consequence of this is that the hedge found by solving the Deep Hedging problem under the statistical measure will be a sum of a true hedge, and a component which does not depend on $Z$ and is simply seeking profitable opportunities in the market. Solving the optimization problem under the risk neutral measure will then directly remove the statistical arbitrage component of the strategy, leaving a ``clean" hedge for the derivative, which is not sensitive to the estimation of the mean returns of our hedging instruments.

\section{Numerical implementations} \label{experiments}

To demonstrate our approach we apply it to two market simulators. First, we discuss a simple, but usable multivariate ``PCA'' vector autoregressive model for spot and a form of implied volatilities. Secondly, we
also present results for a Generative Adversarial Network based simulator based on the ideas presented in~\cite{DHOPT}.

\subsection{Vector Autoregressive market simulator}

For the first numerical experiment, we build a VAR market simulator as follows. For the simulation we use \emph{discrete local volatilies} (DLVs) as  arbitrage-free parametrization of option prices. We do not use the underlying model dynamics; the only use of DLVs is arbitrage-free parametrization of the option surface. We briefly recap the relevant notation:
assume thaty $0=\tau_0<\tau_1<\cdots<\tau_m$ are time-to-maturities and $0<x_1<\ldots<1<\ldots<x_n$  relative strikes.\footnote{See~\cite{DLV} for the use of inhomogeneous strike grids.} 
Also define the additonal boundary strikes $0\leq x_0\ll x_1$ and $x_{n+1}:=1+2 x_n \gg x_n$.

For $i=1,\ldots,n$ and $j=1,\ldots,m$ we denote by $C^{j,i}$ the price of the call option with payoff $(S_{\tau_j}/S_0 - x_i)^+$ at maturity $\tau_j$.
Define $
	\Delta^{j,i} := \frac{ C^{j,i+1} - C^{j,i} }{ x_{i+1} - x_i }$, $\Gamma^{j,i} := \Delta^{j,i} - \Delta^{j,{i-1}}$
and
$
	 \Theta^{j,i} := \frac{ C^{j,i} - C^{j,i-1} }{ \tau_j - \tau_{j-1} }
$.
The \emph{discrete local volatility} surface $(\sigma^{j,i})_{j,i}$ for $j=1,\ldots,m$ and $i=1,\ldots,n$ is defined by
\eq{sigmadef}
	\sigma^{j,i} := \sqrt{ \frac{2\, \Theta^{j,i} }{  x^{j,i}{}^2 \Gamma^{j,i} } }  \ ,
\eqend
where we set $\sigma^{j,i}=\infty$ whenever the square root is imaginary. We also set $\sigma^{j,i} = 0$ if $0/0$ occurs.
We recall that a given surface of option prices is free of static arbitrage if and only if $\sigma<\infty$.
Moreover, given a surface of finite discrete local volatilities, we can reconstruct the surface of arbitrage-free call prices by solving the implicit finite difference scheme implied by~\eqref{sigmadef}.
This involves inverting sequentially $m$ tridiagonal matrices, an operation which is ``on graph" in modern automatic adjoint differentiation (AAD)
machine learning packages such as TensorFlow, see~\cite{DLV} for further details. 

For a given time series of vectors of historical log spot returns and log DLVs
\eqx
	Y_r = \left( \log\frac{ S_r }{ S_{r-1}} , \log \sigma^{1,1}_r, \ldots ,  \log \sigma^{m,n}_r \right )' \ ,
\eqend
we estimate a vector autoregression model of the form 
\eqx
Y_{r} = \left( B - A_{1}Y_{r-1} - A_{2}Y_{r-2} \right) dt +  \sqrt{dt} Z_{r} \ , \ \ \ Z_{r} \sim \calN(0, \Sigma)
\eqend
where each $A_{1,2}$ is a $mn+1 \times mn+1$ coefficient matrix, $B = (B_0, B_1,\ldots, B_{mn} )'$ is an intercept, and $\Sigma$ is a volatility matrix.

\textbf{Constructing~$\P$ --}
we train the model to historical data from EURO STOXX~50, using standard regression techniques from the Statsmodels Python package \cite{STATSMODELS}. Once the model has been trained, we can simulate new sample paths of log spot returns and discrete local volatilities by sampling new noise variables $Z_{r}$ and stepping the model forward. We then convert the DLVs to option prices using the methods detailed above, so that we can simulate market states of spot and option prices.  

We generate $10^5$ paths, of length 30 days from a VAR model, where each path consists of spot and both put and call option prices on a grid of maturities $ \{ 20, 40, 60 \}$ and relative strikes $\{ 0.85, \ldots, 1.15 \}$. 

\textbf{Constructing~$\Q^*$ --}
we construct the risk neutral measure $\Q^*$ by solving \ref{eq:max_me2tc} with the exponential utility, using proportional transaction costs for all instruments set to $\gamma = 0.001$.
 To parametrize our policy action, we use a two layer feedforward neural network, with 64 units in each layer and ReLU activation functions. We train for 2000 epochs on the training set of $10^5$ paths.

\textbf{Assessing Performance --}
Figure~\ref{fig:realised_payoffl} compares out-of-sample the expected value of the option payoffs vs.~their prices for the full grid of calls and puts under both the statistical and the risk-free measure in relation to trading cost.

\begin{figure}[h]
  \centering
\includegraphics[width=0.9\linewidth]{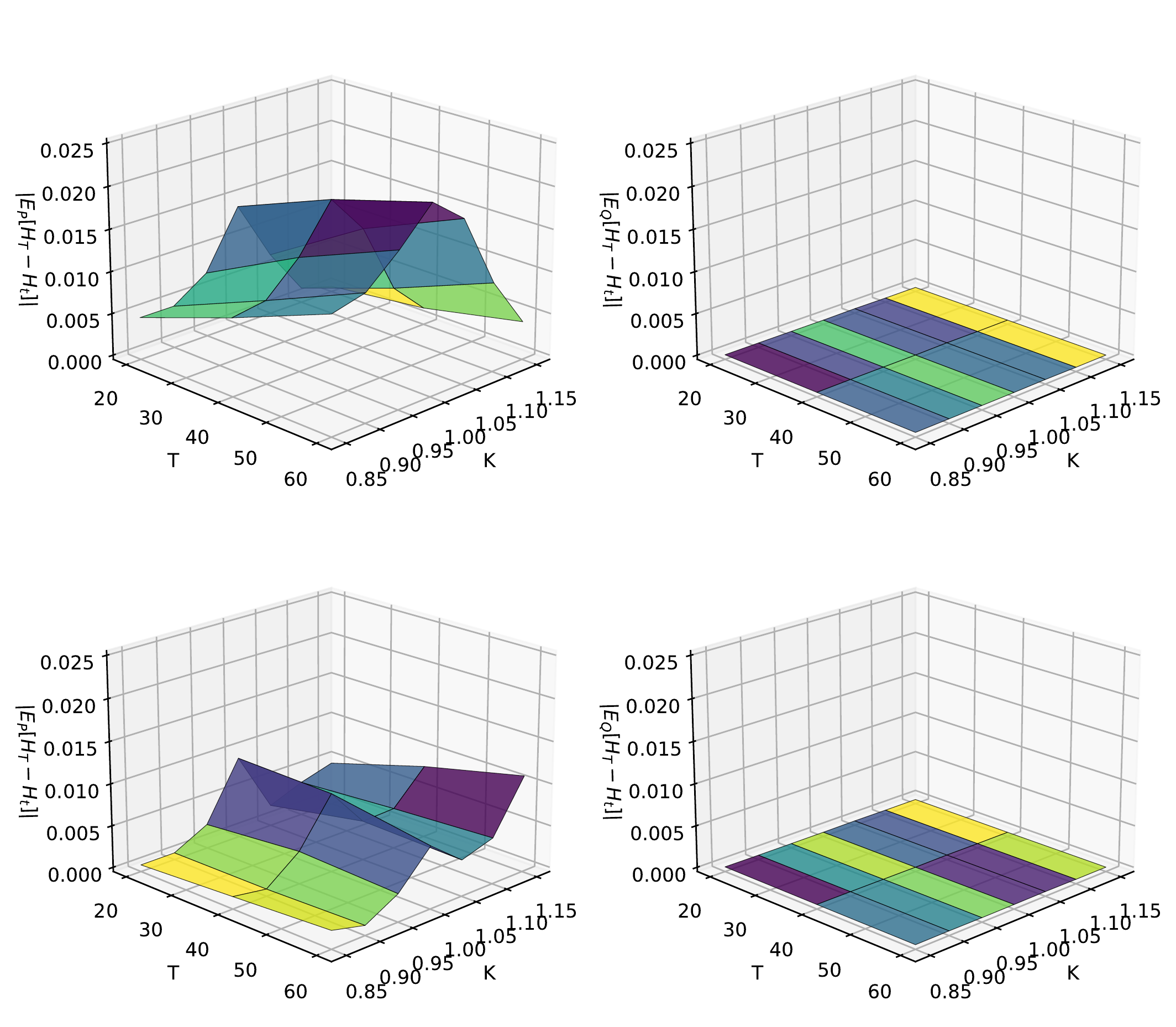}
  \caption{ Average realised drift for call options (top) and put options (bottom) under the $\P$ market simulator (left) and $\Q^*$ simulator (right), by strike and maturity.}
  \label{fig:realised_payoffl}
\end{figure}

The expected payoff under the changed measure has  been flattened to zero, and now lies within the transaction cost level, so that the tradable drift has been removed.

To further confirm that statistical arbitrage has indeed been eliminated from the market simulator under~$\Q^*$, we train a second strategy under the new measure, with identical neural network architecture and the same utility
function but incresed transaction cost~$\gamma = 0.002$. Figure \ref{fig:var_pnl} shows the distributions of terminal gains of respective estimated optimal strategies under $\P$ and $\Q$. We compare the method using the exponential utility, and the adjusted mean-volatility utility $u(x) := ( 1+ \lambda x - \sqrt{1 + \lambda^2 x^2})/\lambda$. In both cases, the distribution of gains is now tightly centred at zero confirming that statistical arbitrage has been removed. 
\begin{figure}[!h]
  \centering
  \includegraphics[width=0.9\linewidth]{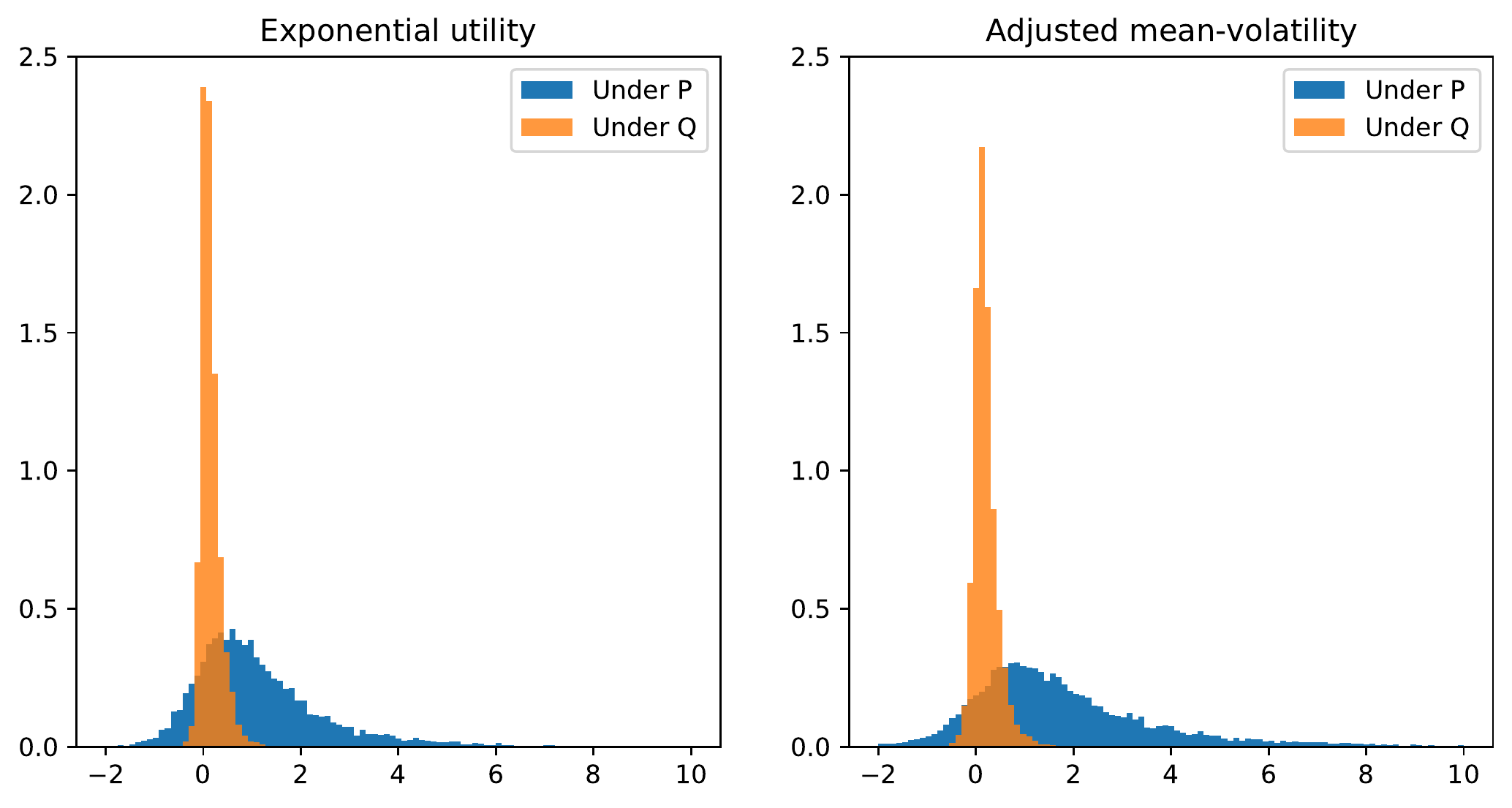}
  \caption{ Gains distribution of estimated optimal policy under $\P$ and $\Q^*$ for exponential utility (left) and adjusted mean-volatility (right).}
  \label{fig:var_pnl}
\end{figure}

\subsection{GAN market simulator}

To demonstrate the flexibility of our approach, we now apply it to a more data driven simulator for spot and option prices based on Generative Adversarial Networks (GANs) \cite{GAN} as in\cite{DHOPT}.

We illustrate the effect on Deep Hedging of changing measure with the following numerical experiment. We hedge a short position in a digital call option, with market instruments being the spot and at the money call options with maturities $20$ and $40$ days. We first train a network under the zero portfolio to find a maximal statistical arbitrage strategy, then use this to construct the risk neutral density. We then train two Deep Hedging networks to hedge the digital, one under the original, unweighted, market and one under the risk neutral market. All networks are trained to maximize exponential utility. Figure \ref{fig:hedge_pnl} compares the final hedged PNL of the two strategies on the left, and the PNL of the strategies, with the statistical arbitrage component subtracted, on the right (i.e. $a_{\Q}^*$ vs. $a_{\P}^* - a_0^*$). The distribution of PNL from the risk neutral hedge is clearly less wide tailed, and the righthand plot demonstrates that we have removed the statistical arbitrage element as the distributions now align.

\begin{figure}
  \centering
\includegraphics[width=0.9\linewidth]{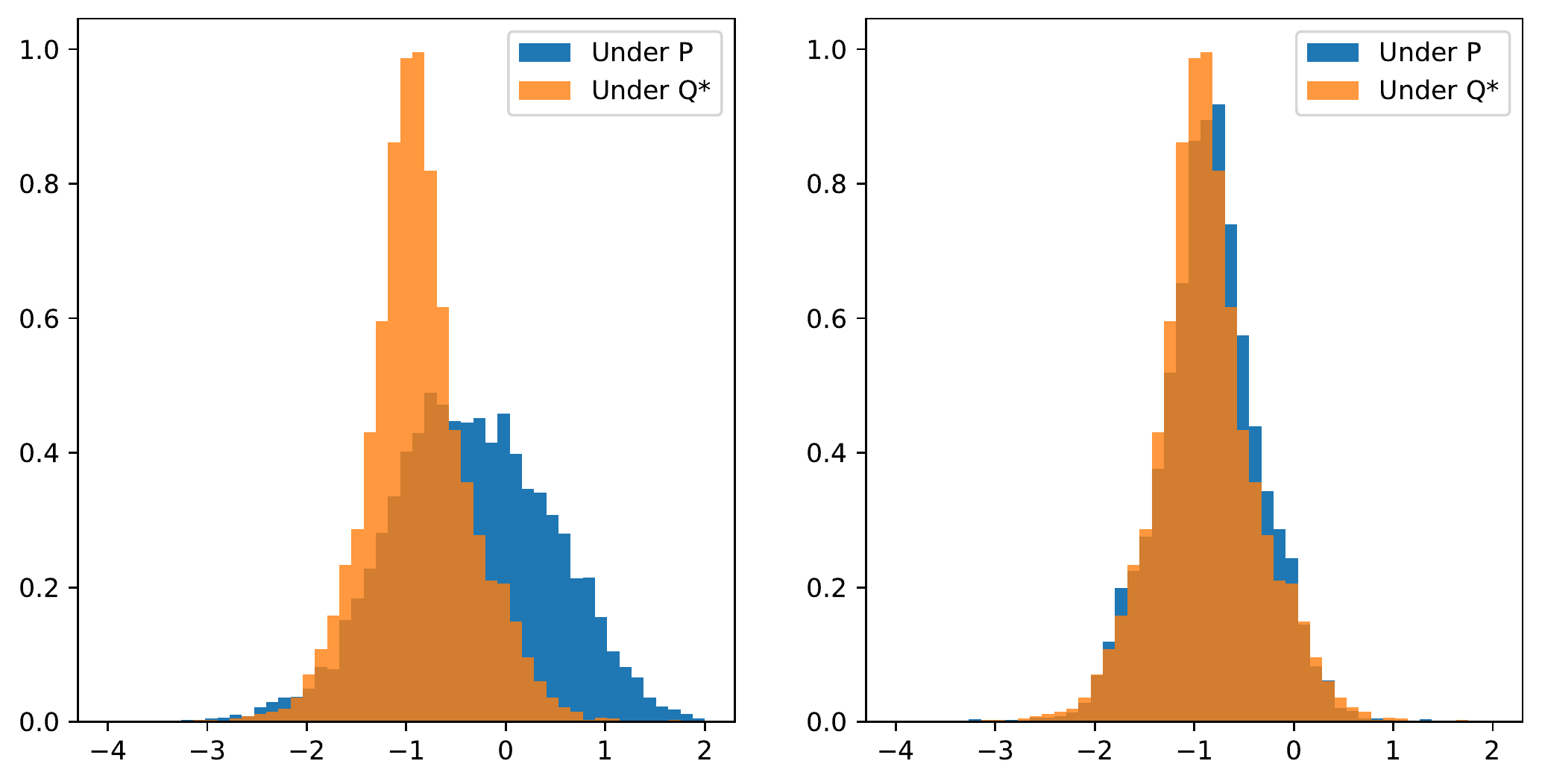}
  \caption{Comparison of the distribution of hedged PNL from a Deep Hedging strategy trained under $\P$ vs one trained under $\Q^*$ (left) and with statistical arbitrage removed (right). Both strategies are evaluated under $\P$. }
  \label{fig:hedge_pnl}
\end{figure}

\textbf{Robustness of~$\Q^*$ --}
by removing the statistical arbitrage component of the strategy, the risk neutral hedge $a_{\Q}^*$ represents a more robust hedge with respect to uncertainty in the market simulator, i.e. when the future distribution of the market at model deployment differs slightly from the distribution of the training data generated by the market simulator. Consider the case where the future market returns follow a distribution $\tilde{\P}$, which is similar to $\P$ in the sense that $H(\tilde{\P} | \P) \leq c$ for some small~$c$ where~$H$ is the relative entropy. For illustration, we can construct such a measure by simply perturbing the weights of the simulated paths slightly.

In particular, we consider perturbations which are unfavourable for the original strategy, where the strategy was over-reliant on the perceived drift in the market, which is no longer present under $\tilde \P$. Figure \ref{fig:robust_hedge_pnl} shows the new PNL distributions under measures $\tilde{\P}$ with $H(\tilde{\P} | \P) = c$ for $c=0.05, 0.5.$ What is striking is that a relatively small shift in the measure can significantly worsen the distribution of hedged PNL of the original Deep Hedging model, but that the distribution of PNL of the model trained on the risk neutral measure remains practically invariant, indicating that the model performances more consistently with respect to uncertainty. Naturally, the method provides robustness against estimation errors for mean returns of the underlying assets.

\begin{figure}
  \centering
\includegraphics[width=0.9\linewidth]{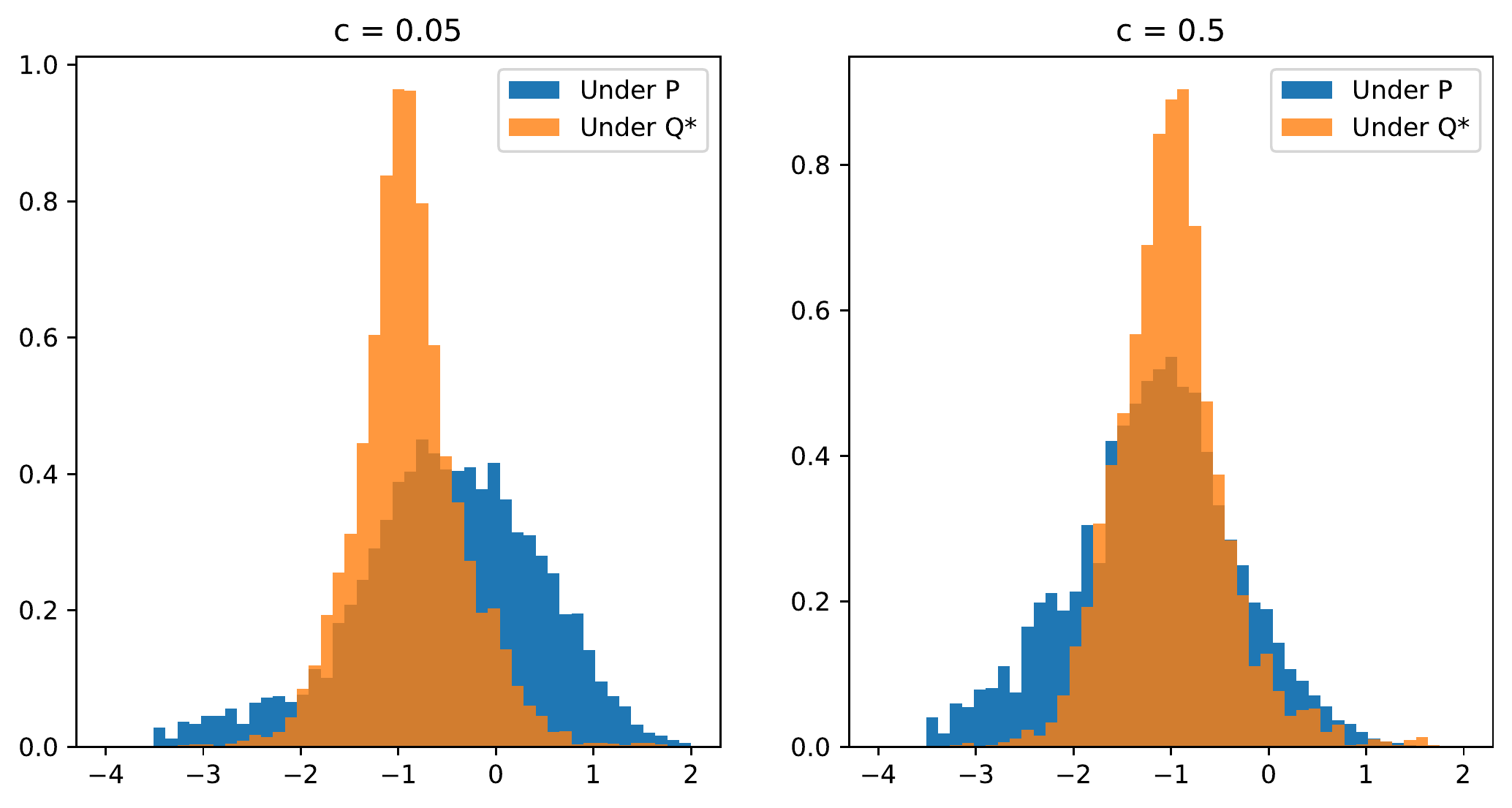}
  \caption{Comparison of hedged PNL from both strategies with respect to uncertainty, evaluated under $\tilde{\P}$ with $H(\tilde{\P} | \P) = c$ for $c=0.05, 0.5$.}
  \label{fig:robust_hedge_pnl}
\end{figure}

\section*{Conclusion}

We have presented a numerically efficient method for computing a risk-neutral density for a set of paths over a number of time steps. Our method is applicable to paths of derivatives and option prices
in particular, hence we effectively provide a framework for statistically learned \emph{stochastic implied volatility} via the application of machine learning tools. 
Our method is generic and does not depend on the market simulator itself, except that it requires that the simulator does not produce classic arbitrage opportunities. It also caters naturally
for transaction costs and trading constraints, and is easily extended to multiple assets.

The method is particularly useful to introduce robustness to a utility-based machine learning approach to the hedging of derivatives, where the use of simulated data is essential to train a `Deep Hedging' neural network model. 
If trained directly on data from the statistical measure, in addition to risk management of the derivative portfolio, the Deep Hedging agent will pursue statistical arbitrage opportunities that appear in the data, thus the hedge action will be polluted by drifts present in the simulated data. By applying our method, we remove any statistical arbitrage opportunities from the simulated data, resulting in a policy from the Deep Hedging agent that seeks to only manage the risk of the derivative portfolio, without exploiting any drifts. This in turn makes the suggested hedge more robust to any uncertainty inherent in the simulated data.


\printbibliography

\section*{Disclaimer}

{\small
Opinions and estimates constitute our judgement as of the date of this Material, are for informational purposes only and are subject to change without notice. It is not a research report and is not intended as such. Past performance is not indicative of future results. This Material is not the product of J.P. Morgan’s Research Department and therefore, has not been prepared in accordance with legal requirements to promote the independence of research, including but not limited to, the prohibition on the dealing ahead of the dissemination of investment research. This Material is not intended as research, a recommendation, advice, offer or solicitation for the purchase or sale of any financial product or service, or to be used in any way for evaluating the merits of participating in any transaction. Please consult your own advisors regarding legal, tax, accounting or any other aspects including suitability implications for your particular circumstances.  J.P. Morgan disclaims any responsibility or liability whatsoever for the quality, accuracy or completeness of the information herein, and for any reliance on, or use of this material in any way.  \\
\noindent
 Important disclosures at: www.jpmorgan.com/disclosures}


\end{document}